\documentclass[runningheads]{llncs}
\pdfoutput=1
\usepackage[utf8]{inputenc}
 \usepackage[numbers,sectionbib]{natbib}
\usepackage{amsmath,verbatim,amssymb,amsfonts,amscd, graphicx}
\usepackage{algorithmicx}
\usepackage[noend]{algpseudocode}
\usepackage{algorithm}
\usepackage{graphics}
\usepackage{mathtools}
\usepackage{lscape}
\usepackage{todonotes}
\usepackage{float}

\newcommand{\lm}{lmw}

\newcommand{\LM}{LMIM-width\ }
\newcommand{\lmo}{mim}

\begin{document}
\title{Linear MIM-Width of Trees
\thanks{This is the appendix of our WG submission, the long version with extra figures and full proofs}}

\author{Svein Høgemo \and Jan Arne Telle  \and Erlend Raa Vågset}
\authorrunning{S. Høgemo et al.} 
\institute{Department of Informatics, University of Bergen, Norway. \\
\email{\{svein.hogemo, jan.arne.telle, erlend.vagset\}@uib.no}}

\maketitle
\begin{abstract}
We provide an $O(n \log n)$ algorithm computing the linear maximum induced matching width of a tree and an optimal layout. 
\end{abstract}

\section{Introduction}
The study of structural graph width parameters like tree-width, clique-width and rank-width has been ongoing for a long time, and their algorithmic use has been steadily increasing \cite{HlinenyOSG08, Oum17}. The maximum induced matching width, denoted MIM-width, and the linear variant LMIM-width, are graph parameters having very strong modelling power introduced by Vatshelle in 2012 \cite{vatshelle2012new}. 
The LMIM-width parameter asks for a linear layout of vertices such that the bipartite graph induced by 
edges crossing any vertex cut has a maximum induced matching of bounded size.
Belmonte and Vatshelle \cite{BELMONTE201354} \footnote{In \cite{BELMONTE201354}, results are stated in terms of $d$-neighborhood equivalence, but in the proof, they actually gave a bound on LMIM-width.} showed that {\sc interval} graphs, {\sc bi-interval} graphs, {\sc convex} graphs and {\sc permutation} graphs, where clique-width can be proportional to the square root of the number of vertices \cite{GolumbicR99}, all have LMIM-width $1$ and an optimal layout can be found in polynomial time. 

Since many well-known classes of graphs have bounded MIM-width or LMIM-width, algorithms that run in 
XP time in these parameters will yield polynomial-time algorithms 
on several interesting graph classes at once. Such algorithms have been developed for many problems: by Bui-Xuan et al ~\cite{BUIXUAN201366} for the class of LCVS-VP - Locally Checkable Vertex Subset and Vertex Partitioning - problems, by Jaffke et al for non-local problems like {\sc Feedback Vertex Set} \cite{JaffkeKT18, JaffkeKT17} and also for {\sc Generalized Distance Domination} 
\cite{JaffkeKST18},  by Golovach et al \cite{GolovachHKKSV18} for output-polynomial {\sc Enumeration of Minimal Dominating} sets, by Bergougnoux and Kant\'e \cite{Bergougnoux} for several Connectivity problems
and by Galby et al for {\sc Semitotal Domination} \cite{abs-1810-06872}. These results give a common explanation for many classical results in the field of algorithms on special graph classes and extends them to the field of parameterized complexity.

Note that very low MIM-width or LMIM-width still allows quite complex cuts compared to similarly defined graph parameters. For example, carving-width $1$ allows just a single edge, maximum matching-width $1$ a star graph, and rank-width $1$ 
a complete bipartite graph. In contrast, LMIM-width $1$ allows any cut where the neighborhoods of the vertices in a color class can be ordered linearly w.r.t. inclusion. 
In fact, it is an open problem whether the class of graphs having LMIM-width $1$ can be recognized in polynomial-time or if this is NP-complete.
S\ae ther et al \cite{SaetherV16} showed that computing the exact MIM-width and LMIM-width of general graphs is W-hard  and not in APX unless NP=ZPP, while Yamazaki \cite{Yamazaki18} shows that under the small set expansion hypothesis it is not in APX unless P=NP. The only graph classes where we know an exact polynomial-time algorithm computing LMIM-width are the above-mentioned classes {\sc interval}, {\sc bi-interval}, {\sc convex} and {\sc permutation} that all have structured neighborhoods implying LMIM-width $1$ \cite{BELMONTE201354}.
Belmonte and Vatshelle also gave polynomial-time algorithms showing that {\sc circular arc} and {\sc circular permutation} graphs have LMIM-width at most $2$, while
{\sc Dilworth} $k$ and $k$-{\sc trapezoid} have LMIM-width at most $k$ \cite{BELMONTE201354}. 
Recently, Fomin et al
\cite{FGR18} showed that LMIM-width for the very general class of {\sc $H$-graphs} is bounded by $2|E(H)|$, and that a layout can be found in polynomial time if given an $H$-representation of the input graph. However, none of these results compute the exact LMIM-width. On the negative side, Mengel \cite{Mengel16} has shown that {\sc strongly chordal split} graphs, {\sc co-comparability} graphs and
{\sc circle} graphs all can have MIM-width, and LMIM-width, linear in the number of vertices.

Just as LMIM-width can be seen as the linear variant of MIM-width, path-width can be seen as the linear variant of tree-width. Linear variants of other well-known parameters like clique-width and rank-width have also been studied. Arguably, the linear variant of MIM-width commands a more noteworthy position, since in contrast to these other linear parameters, for almost all well-known graph classes where the original parameter (MIM-width) is bounded then also the linear variant (LMIM-width) is bounded. 

In this paper we give an $O(n \log n)$ algorithm computing the LMIM-width of an $n$-node tree. 
This is the first graph class of LMIM-width larger than $1$ having a polynomial-time algorithm computing LMIM-width and thus constitutes an important step towards a better understanding of this parameter.
The path-width of trees was first studied in the early 1990s by M\"ohring \cite{mohring1990graph}, with Ellis et al \cite{EllisST94} giving an $O(n \log n)$ algorithm computing an optimal path-decomposition, and Bodlaender \cite{Bod96} an $O(n)$ algorithm. 
In 2013
Adler and Kant\'e \cite{AdlerK15} gave linear-time algorithms computing the linear rank-width of trees and also the linear clique-width of trees, by reduction to the path-width algorithm. 
Even though LMIM-width is very different from path-width, the basic framework of our algorithm is similar to the  path-width algorithm in \cite{EllisST94}.

In Section \ref{sec:class} we give some standard definitions and prove the Path Layout Lemma, that if a tree $T$ has a path $P$ such that all components of $T \setminus N[P]$ have LMIM-width at most $k$ then $T$ itself has a linear layout with LMIM-width at most $k+1$.
We use this to prove a classification theorem stating that a tree $T$ has LMIM-width at least $k+1$ if and only if there is a node
$v$ such that after rooting $T$ in $v$, at least three children of $v$ {\em themselves have at least one child} whose rooted subtree has LMIM-width at least $k$. From this it follows that the LMIM-width of an $n$-node tree is no more than $\log n$. 
Our $O(n \log n)$ algorithm computing LMIM-width of a tree $T$ picks an arbitrary root $r$ and proceeds bottom-up on the rooted tree $T_r$. In Section \ref{sec:root} we show how to assign labels to the rooted subtrees encountered in this process giving their LMIM-width. However, as with the algorithm computing pathwidth of a tree, the label is sometimes complex, consisting of LMIM-width of a sequence of subgraphs, of decreasing LMIM-width, that 
are not themselves full rooted subtrees.
Proposition \ref{Prop: Case} is an 8-way case analysis giving a subroutine used to update the label at a node given the labels at all children.
In Section \ref{sec:comp} we give our bottom-up algorithm, which will make calls to the subroutine underlying Proposition \ref{Prop: Case} in order to compute the complex labels and the LMIM-width. 
%
Finally, we use all the computed labels to lay out the tree in an optimal manner.

\section{Classifying LMIM-width of Trees} \label{sec:class}

We use standard graph theoretic notation, see e.g. \cite{Diestel}.
For a graph $G=(V,E)$ and subset of its nodes $S \subseteq V$ we denote by $N(S)$ the set of neighbors of nodes in $S$, by $N[S]= S \cup N(S)$ its closed neighborhood, and by $G[S]$ the graph induced by $S$.
For a bipartite graph $G$ we denote by MIM(G), or simply MIM if the graph is understood, the size of its Maximum Induced Matching, the largest number of edges whose endpoints induce a matching.
Let $\sigma$ be the linear order corresponding to the enumeration $v_1, \dots, v_n$ of the nodes of $G$, this will also be called a linear layout of $G$. 
For any index $1 \leq  i < n$ we have a cut of $\sigma$ that defines the bipartite graph on edges ''crossing the cut'' i.e. edges with one endpoint in $\{v_1, \dots, v_i\}$ and the other endpoint in $\{v_{i+1}, \dots, v_n\}$. 
The maximum induced matching of $G$ under layout $\sigma$ is denoted $\lmo(\sigma,G)$, and is defined as the maximum, over all cuts of $\sigma$, of the value attained by the MIM of the cut, i.e. of the bipartite graph defined by the cut. 
%
The linear induced matching width -- LMIM-width -- of $G$ is denoted $\lm(G)$, and is the minimum value of $\lmo(\sigma,G)$ over all possible linear orderings $\sigma$ of the vertices of $G$.

We start by showing that if we have a path $P$ in a tree $T$ then the LMIM-width of $T$ is no larger than the largest LMIM-width of any component of $T \setminus N[P]$, plus  $1$ . To define these components the following notion is useful.

\begin{definition}[Dangling tree]
Let $T$ be a tree containing the adjacent nodes $v$ and $u$. The dangling tree from $v$ in $u$, $T\langle v,u\rangle$, is the component of $T \setminus (u,v)$ containing $u$.
\end{definition}

Given a node $x \in T$ with neighbours $\{v_1,\ldots,v_d\}$, the forest obtained by removing $N[x]$ from $T$ is a collection of dangling trees $\{T\langle v_i,u_{i,j}\rangle\}$, where $u_{i,j} \neq x$ is some neighbour of $v_i$. We can generalise this to a path $P = (x_1,\ldots,x_p)$ in place of $x$, such that $T\backslash N[P] = \{T\langle v_{i,j},u_{i,j,m}\rangle\}$, where $v_{i,j} \in N(P)$ is a neighbour of $x_i$ and $u_{i,j,m} \not\in N[P]$. See top part of Figure \ref{fig-lay1}. This naming convention will be used in the following.

\begin{lemma}[Path Layout Lemma] \label{Path Layout Lemma}
Let $T$ be a tree. If there exists a path $P=(x_1,\ldots,x_p)$ in $T$ such that every connected component of $T\backslash N[P]$ has \LM $\leq k$ then $\lm(T) \leq k+1$. Moreover, given the layouts for the components we can in linear time compute the layout for $T$.
\end{lemma}

\begin{proof}
Using the optimal linear orderings of the connected components of $T\backslash N[P]$, we give the below algorithm \textsc{LinOrd} constructing a linear order $\sigma_T$ on the nodes of $T$ showing that \lm  of $T$ is $\leq k+1$. The ordering $\sigma_T$ starts out empty and the algorithm has an outer loop going through vertices in the path $P=(x_1,\ldots,x_p)$. When arriving at $x_i$ it uses the concatenation operator $\oplus$ to add the path node $x_i$ before looping over all neighbors $v_{i,j}$ of $x_i$ adding the linear orders of each dangling tree from $v_{i,j}$ and then $v_{i,j}$ itself.
See Figure \ref{fig-lay1} for an illustration.

\begin{algorithmic}
\Function{LinOrd}{$T$: tree, $P=(x_1,\ldots,x_p)$: path,
$\{\sigma_{T\langle v_{i,j}, u_{i,j,m}\rangle}\}$: lin-ords} 
\State $\sigma_T \gets \emptyset$ \Comment {The list starts out empty}
\For {$i\gets 1,p$} \Comment {For all nodes on path $(x_1,\ldots,x_p)$} 
	\State $\sigma_T \gets \sigma_T \oplus x_i$ \Comment {Append path node}
	\For {$j\gets 1,|N(x_i)\backslash P|$} \Comment {For all nbs of $x_i$ not on path: $v_{i,j}$}
		\For {$m\gets 1,|N(v_{i,j})\backslash x_i|$} \Comment{For all dangling trees from $v_{i,j}$} 
			\State $\sigma_T \gets \sigma_T \oplus 
				\sigma_{T\langle v_{i,j}, u_{i,j,m}\rangle}$ \Comment {Append given order of $T\langle v_{i,j}, u_{i,j,m}\rangle$}
		\EndFor
		\State $\sigma_T \gets \sigma_T \oplus v_{i,j}$ \Comment {Append $v_{i,j}$}
	\EndFor
\EndFor
\EndFunction
\end{algorithmic}

\begin{center}
\begin{figure}[ht!]
\centering
\includegraphics[width=100mm]{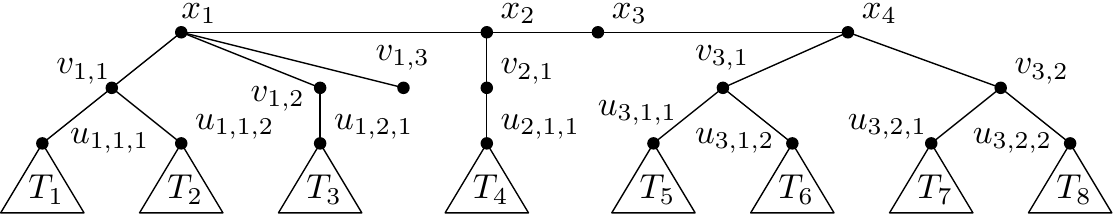}

\bigskip

\includegraphics[width=120mm]{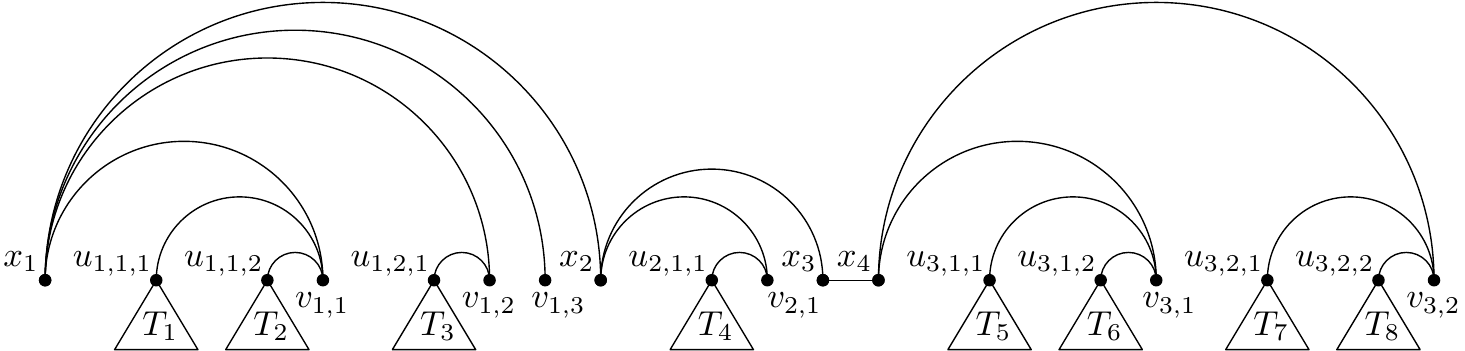}
\caption{A tree with a path $P=(x_1,x_2,x_3,x_4)$, with nodes in $N[N[P]]$ and dangling trees featured, and below it the order given by the Path Layout Lemma}
\label{fig-lay1}
\end{figure}
\end{center}

Firstly, from the algorithm it should be clear that each node of $T$ is added exactly once to $\sigma_T$, that it runs in linear time, and that there is no cut containing two crossing edges from two separate dangling trees. 
Now we must show that $\sigma_T$ does not contain cuts with MIM larger than $k+1$. By assumption the layout of each dangling tree has no cut with MIM larger than $k$, and since these layouts can be found as subsequences of $\sigma_T$ it follows that then also $\sigma_T$ has no cut with more than $k$ edges from a single dangling tree $T\langle v_{i,j}, u_{i,j,m}\rangle$. Also, we know that edges from two separate dangling trees cannot both cross the same cut. The only edges of $T$ left to account for, i.e. not belonging to one of the dangling trees, are those with both endpoints in $N[N[P]]$, the nodes at distance at most 2 from a node in $P$. 
For every cut of $\sigma_T$ that contains more than a single crossing edge $(x_i,x_{i+1})$ there is a unique $x_i \in P$ and a unique $v_{i,j} \in N(x_i)$ such that every edge with both endpoints in $N[N[P]]$ that crosses the cut is incident on either $x_i$ or $v_{i,j}$, and since the edge connecting $x_i$ and $v_{i,j}$ also crosses the cut at most one of these edges can be taken into an induced matching.
With these observations in mind, it is clear that $\lm(T) \leq \lmo(\sigma_T,T) \leq k+1$.

\end{proof}

\begin{definition}[$k$-neighbour and $k$-component index]
Let $x$ be a node in the tree $T$ and $v$ a neighbour of $x$. If $v$ has a neighbour $u \neq x$ such that $\lm(T\langle v,u\rangle) \geq k$, then we call $v$ a \textbf{$k$-neighbour} of $x$. The \textbf{$k$-component index of $x$} is equal to the number of $k$-neighbours of $x$ and is denoted $D_T(x,k)$, or shortened to $D(x,k)$.
\end{definition}

\begin{theorem}[Classification of LMIM-width of Trees] \label{Thm: Classification theorem}
For a tree $T$ and $k \geq 1$ we have $\lm(T) \geq k+1$ if and only if $D(x,k) \geq 3$ for some node  $x$.
\end{theorem}

\begin{proof}
We first prove the backward direction by contradiction. Thus we assume $D(x,k) \geq 3$ for a node $x$ and there is a linear order $\sigma$ such that $\lmo(\sigma,T) \leq k$. 

Let $v_1,v_2,v_3$ be the three $k$-neighbors of $x$ and $T_1,T_2,T_3$ the three trees of $T\setminus N[x]$ each of LMIM-width $k$, with $v_i$ connected to a node of $T_i$ for $i=1,2,3$, that we know must exist by the definition of $D(x,k)$. 
We know that for each $i=1,2,3$ we have a cut $C_i$ in $\sigma$ with MIM=$k$ and all $k$ edges of this induced matching coming from the tree $T_i$. Wlog we assume these three cuts come in the order $C_1, C_2, C_3$, i.e. with the cut having an induced matching of $k$ edges of $T_2$ in the middle. Note that in $\sigma$ all nodes of $T_1$ must appear before $C_2$ and all nodes of $T_3$ after $C_2$, as otherwise, since $T$ is connected and the distance between $T_2$ and the two trees $T_1$ and $T_3$ is at least two, there would be an extra edge crossing $C_2$ that would increase MIM of this cut to $k+1$. 
It is also clear that $v_1$ has to be placed before $C_2$ and $v_3$ has to be placed after $C_2$, for the same reason, e.g. the edge between $v_1$ and a node of $T_1$ cannot cross $C_2$ without increasing MIM. But then we are left with the vertex $x$ that cannot be placed neither before $C_2$ nor after $C_2$ without increasing MIM of this cut by adding at least one of $(v_1,x)$ or $(v_3,x)$ to the induced matching. We conclude that $D(x,k) \geq 3$ for a node $x$ implies LMIM-width at least $k+1$.
%


To prove the forward direction we first show the following partial claim:
if $\lm(T) \geq k+1$ then there exists a node $x \in T$ such that $D(x,k)\geq 3$; or there exists a strict subtree $S$ of $T$ with $\lm(S) \geq k+1$.
We will prove the contrapositive statement, so let us assume that every node in $T$ has $D(x,k)<3$ and no strict subtree of $T$ has \LM $\geq k+1$ and show that then $\lm(T) \leq k$.
For every node $x \in T$, it must then be true that $D(x,k) \leq 2$ and that $D(x,k+1) = 0$. The strategy of this proof is to show that there is always a path $P$ in $T$ such that all the connected components in $T \backslash N[P]$ have \LM $\leq k-1$. When we have shown this, we proceed to use the Path Layout Lemma, to get that $\lm(T) \leq k$.
To prove this, we define the following two sets of vertices: 
$$X = \{x|x\in V(T)\ and\ D(x,k)=2\}, Y = \{y|y\in V(T)\ and\ D(y,k)=1\}$$
Case 1: $X \neq \emptyset$\\
If $x_i$ and $x_j$ are in $X$, then every vertex on the path $P(x_i,\ldots,x_j)$ connecting $x_i$ and $x_j$ must be elements of $X$, as every node on this path clearly has a dangling tree with \LM $k$ in the direction of $x_i$ and in the direction of $x_j$. The fact that every pair of vertices in $X$ are connected by a path in $X$ means that $X$ must be a connected subtree of $T$.
Furthermore, this subtree must be a path, otherwise there are three disjoint dangling trees $T\langle  v_1,u_1\rangle, T\langle v_2,u_2\rangle, T\langle v_3, u_3\rangle$, each with \LM $k$, and each hanging from a separate node.
But then there is some vertex $w$ such that $T\langle  v_1,u_1\rangle, T\langle v_2,u_2\rangle $ and $ T\langle v_3, u_3\rangle$ are subtrees of dangling trees from different neighbours of $w$. But this implies that $D(w,k)\geq 3$, which we assumed were not the case, so this leads to a contradiction. We therefore conclude that all nodes in $X$ must lie on some path $P=(x_1,\ldots,x_p)$.
The final part of the argument lies in showing that we can apply the Path Layout Lemma. For some $x_i \in P, i \in \{2,\ldots,p-1\}$, its $k$-neighbours are $x_{i-1}$ and $x_{i+1}$. For $x_1$, these neighbours are $x_2$ and some $x_0 \not\in X$. For $x_p$, these neighbours are $x_{p-1}$ and some $x_{p+1} \not\in X$.
$x_0$ and $x_{p+1}$ may only have one $k$-neighbour – $x_1$ and $x_p$ respectively – or else they would be in $X$. If we make $P' = (x_0,\ldots,x_{p+1})$, we then see that every connected component in $T\backslash N[P']$ must have \LM $\leq k-1$. By the Path Layout Lemma, $\lm(T) \leq k$.
\\

Case 2: $X= \emptyset$, $Y \neq \emptyset$\\
We construct the path $P$ in a simple greedy manner as follows. We start with $P=(y_1,y_2)$, where $y_1$ is some arbitrary node in $Y$, and $y_2$ its only $k$-neighbour. Then, if the highest-numbered node in $P$, call it $y_q$, has a $k$-neighbour $y' \not\in P$, then we assign $y_{q+1}$ to $y'$, and repeat this process exhaustively.
Since we look at finite graphs, we will eventually reach some node $y_p$ such that either $y_p \not\in Y$ or $y_p$'s $k$-neighbour is $y_{p-1}$. We are then done and have $P=(y_1,\ldots,y_p)$, which must be a path in $T$, since every node $y_{i+1}\in P$ is a neighbour of $y_i$ and for $y_i$ we only assign maximally one such $y_{i+1}$. Also, every connected component of $T\backslash N[P]$ must have \LM $\leq k-1$. If not, some node $y_i\in P$ would have a $k$-neighbour $y' \not\in P$, but by the assumption $X= \emptyset$ this is impossible, since then 
either $i<p$ and $y_i$ has two $k$-neighbours $y'$ and $y_{i+1}$, or else
$i=p$ and $y_p\in Y$ and $y_i$ has the two $k$-neighbors $y'$ and $y_{i-1}$ (in case
$i=p$ and $y_p\not\in Y$ then by definition of $Y$ the node $y_i$ could not have a $k$-neighbor $y'$). By the Path Layout Lemma, $\lm(T) \leq k$. \\

Case 3: $X = \emptyset$, $Y= \emptyset$\\
If you make $P=(x)$ for some arbitrary $x\in T$, it is obvious that every connected component of $T \backslash N[P]$ has \LM $\leq k-1$. By the Path Layout Lemma, $\lm(T) \leq k$.
\\

We have proven the partial claim that if $\lm(T) \geq k+1$ then there exists a node $x \in T$ such that $D(x,k)\geq 3$; or there exists a strict subtree $S$ of $T$ with $\lm(S) \geq k+1$.
To finish the backward direction of the theorem we need to show that
if $\lm(T) \geq k+1$ then there exists a node $x \in T$ with $D(x,k) \geq 3$.
Assume for a contradiction that there is no node with $k$-component index at least 3 in $T$. By the partial claim, there must then exist a strict subtree $S$ with $\lm(S) \geq k+1$. But since we look at finite trees, we know that there in $S$ must exist a minimal subtree $S_0, \lm(S_0) = k+1$ with no strict subtree with \LM $> k$. By the partial claim, $S_0$ must contain a node $x_0$ with $D_{S_0}(x_0,k) \geq 3$. But every dangling tree $S_0\langle v,u\rangle$ is a subtree of $T\langle v,u\rangle$, and so if $D_{S_0}(x_0,k) \geq 3$, then $D_T(x_0,k) \geq 3$ contradicting our assumption.
\end{proof}

\begin{center}
\begin{figure}[ht!]
\centering
\includegraphics[width=70mm]{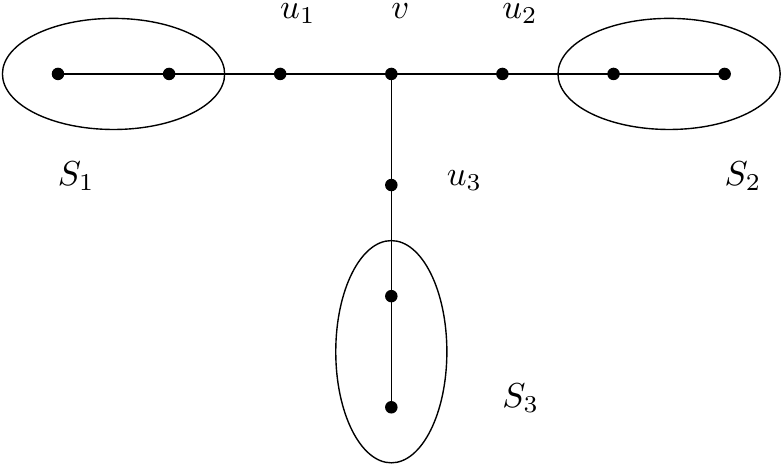}
\caption{The smallest tree with LMIM-width 2, having a node $v$ with three 1-neighbors $u_1,u_2,u_3$ having dangling trees $S_1,S_2,S_3$, respectively, so that $D(v,1)=3$ }
\end{figure}
\end{center}

By Theorem \ref{Thm: Classification theorem}, every tree with \LM $k \geq 2$ must be at least 3 times bigger than the smallest tree with \LM $k-1$, which implies the following.

\begin{remark} \label{Rmk: log lmw}
The \LM of an $n$-node tree is $\mathcal{O}(log\ n)$.
\end{remark}

\section{Rooted trees, $k$-critical nodes and labels} \label{sec:root}

Our algorithm computing LMIM-width will work on a rooted tree, processing it bottom-up.
We will choose an arbitrary node $r$ of the tree $T$ and denote by $T_r$ the tree rooted in $r$. 
For any node $x$ we denote by $T_r[x]$ the standard \emph{complete subtree} of $T_r$ rooted in $x$. During the bottom-up processing of $T_r$ we will compute a label for various subtrees. The notion of a $k$-critical node is crucial for the definition of labels.

\begin{definition}[$k$-critical node]
Let $T_r$ be a rooted tree with $\lm(T_r) = k$. We call a node $x$ in $T_r$ \textbf{$k$-critical} if it has exactly two children $v_1$ and $v_2$ that each has at least one child, $u_1$ and $u_2$ respectively, such that $\lm(T_r[u_1]) = \lm(T_r[u_2]) = k$. Thus $x$ is $k$-critical if and only if $\lm(T)=k$ and $D_{T_r[x]}(x,k)=2$.
\end{definition}


\begin{remark} \label{Rmk: single k-crit} 
If $T_r$ has LMIM-width $k$ it has at most one $k$-critical node.
\end{remark}

\begin{proof}
For a contradiction, let $x$ and $x'$ be two $k$-critical nodes in $T_r$. There are then four nodes, $v_l,v_r,v_l',v_r'$, the two $k$-neighbours of $x$ and $x'$ respectively, such that there exist dangling trees $T\langle v_l,u_l\rangle, T\langle v_r,u_r\rangle, T\langle v_l',u_l'\rangle, T\langle v_r',u_r'\rangle$ that all have \LM $k$. 
If $x$ and $x'$ have a descendant/ancestor relationship in $T_r$, then assume wlog that $x'$ is a descendant of $v_l$, and note that $T\langle v_r,u_r\rangle, T\langle v_l',u_l'\rangle$ and $T\langle v_r',u_r'\rangle$ are disjoint trees in different neighbours of $x'$, thus $D_{T_r}(x',k)=3$ and by Theorem \ref{Thm: Classification theorem} $T_r$ should have \LM $k+1$
Otherwise, all the dangling trees are disjoint, thus $D_T(x,k) = D_T(x',k) = 3$ and we arrive at the same conclusion.
\end{proof}

\begin{definition}[label]\label{Def:label}
Let rooted tree $T_r$ have $\lm(T_r) = k$. Then $\mathbf{label(T_r)}$ consists of a list of decreasing numbers, $(a_1,\ldots,a_p)$, where $a_1 = k$, appended with a string called $last\_type$, which tells us where in the tree an $a_p$-critical node lies, if it exists at all. If $p=1$ then the label is simple, otherwise it is complex. The $\mathbf{label(T_r)}$ is defined recursively, with type 0 being a base case for singletons and for stars, and with type 4 being the only one defining a complex label.

\begin{itemize}
\item{Type 0:} $r$ is a leaf, i.e. $T_r$ is a singleton, then $label(T_r) = (0,t.0)$; \\
or all children of $r$ are leaves, then $label(T_r) = (1,t.0)$
\item{Type 1:} No $k$-critical node in $T_r$, then $label(T_r) = (k,t.1)$
\item{Type 2:} $r$ is the $k$-critical node in $T_r$, then $label(T_r) = (k,t.2)$
\item{Type 3:} A child of $r$ is $k$-critical in $T_r$, then $label(T_r) = (k,t.3)$
\item{Type 4:} There is a $k$-critical node $u_k$ in $T_r$ that is neither $r$ nor a child of $r$. Let $w$ be the parent of $u_k$. Then $label(T_r) = k \oplus label(T_r\backslash T_r[w])$ 
\end{itemize}
\end{definition}

In type 4 we note that $\lm(T_r\backslash T_r[w]) < k$ since otherwise $u_k$ would have three $k$-neighbors (two children in the tree and also its parent) and by Theorem \ref{Thm: Classification theorem} we would then have $\lm(T_r) = k+1$. Therefore, all numbers in $label(T_r\backslash T_r[w])$ are smaller than $k$ and a complex label is a list of decreasing numbers followed by $last\_type \in \{t.0, t.1, t.2, t.3\}$. We now give a Proposition that for any node $x$ in $T_r$ will be used to compute $label(T_r[x])$ based on the labels of the subtrees rooted at the children and grand-children of $x$. The subroutine underlying this Proposition, see the decision tree in Figure \ref{diagram}, will be used when reaching node $x$ in the bottom-up processing of $T_r$.

\begin{proposition} \label{Prop: Case}
Let $x$ be a node of $T_r$ with children $Child(x)$, and given $label(T_r[v])$
for all $v \in Child(x)$. We define (and compute) $k = max_{v \in Child(x)}\ \{\lm(T_r[v])\}$ and $N_k = \{v \in Child(x)\ |\ \lm(T[v]) = k\}$ and denote by $N_k=\{v_1,\ldots,v_q\}$ and by $l_i = label(T_r[v_i])$.
Define (compute) $t_k = D_{T_r[x]}(x,k)$ by noting that $t_k = |\{v_i \in N_k\ |\ v_i\ \mathrm{has\ child}\ u_j\ \mathrm{with}\ \lm(T_r[u_j]) = k\}|$. Given this information, we can find $label(T_r[x])$ as follows:
\begin{itemize}
\item \textbf{Case 0:} if $|Child(x)| = 0$ then $label(T_r[x]) = (0,t.0)$; \\
else if $k = 0$ then $label(T_r[x]) = (1,t.0)$
\item \textbf{Case 1:} Every label in $N_k$ is simple and has $last\_type$ equal to $t.1$ or $t.0$, and $t_k \leq 1$. Then, $label(T_r[x]) = (k,t.1)$
\item \textbf{Case 2:} Every label in $N_k$ is simple and has $last\_type$ equal to $t.1$ or $t.0$, but $t_k = 2$. Then, $label(T_r[x]) = (k,t.2)$
\item \textbf{Case 3:} Every label in $N_k$ is simple and has $last\_type$ equal to $t.1$ or $t.0$, but $t_k \geq 3$. Then, $label(T_r[x]) = (k+1,t.1)$
\item \textbf{Case 4:} $|N_k| \geq 2$ and for some $v_i \in N_k$, either $l_i$ is a complex label, or $l_i$ has $last\_type$ equal to either $t.2$ or $t.3$. Then, $label(T_r[x]) = (k+1,t.1)$
\item \textbf{Case 5:} $|N_k| = 1$, $l_1$ is a simple label and $l_1$ has $last\_type$ equal to $t.2$. Then, $label(T_r[x]) = (k,t.3)$
\item \textbf{Case 6:} $|N_k| = 1$, $l_1$ is either complex or has $last\_type$ equal to $t.3$, and $k \not \in label(T_r[x]\backslash T_r[w])$, where $w$ is the parent of the $k$-critical node in $T_r[v_1]$. Then, $label(T_r[x]) = k \oplus label(T_r[x]\backslash T_r[w])$
\item \textbf{Case 7:} $|N_k| = 1$, $l_1$ is either complex or has $last\_type$ equal to $t.3$, and $k \in label(T_r[x]\backslash T_r[w])$, where $w$ is the parent of the $k$-critical node in $T_r[v_1]$. Then, $label(T_r[x]) = (k+1,t.1)$
\end{itemize}
\end{proposition}

\begin{center}
\begin{figure}[ht!]
\centering
\includegraphics[width=120mm]{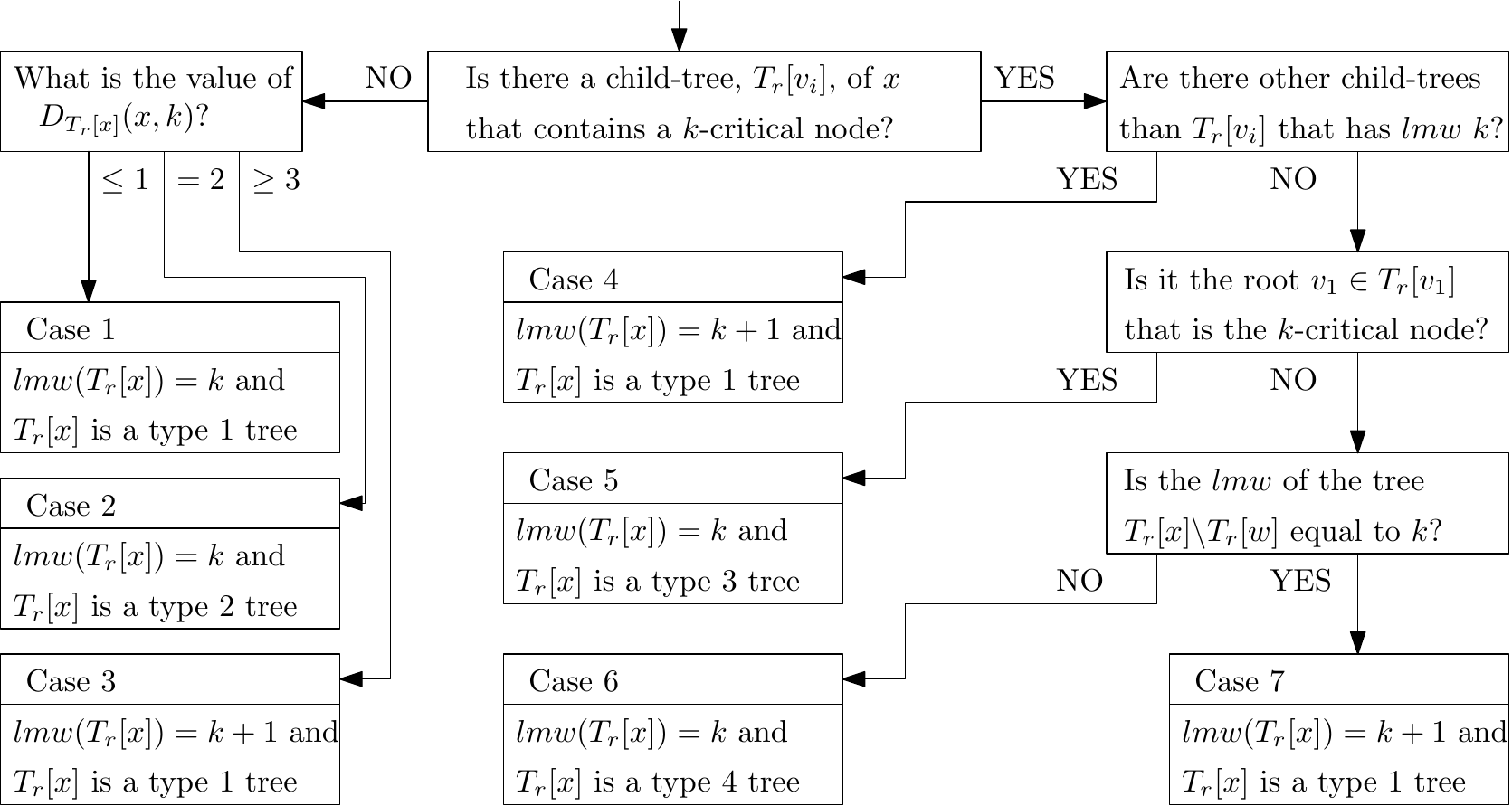}
\caption{A decision tree corresponding to the case analysis of Proposition \ref{Prop: Case}}
\label{diagram}
\end{figure}
\end{center}
\begin{proof}
We show that exactly one case applies to every rooted tree and in each case we assign the label according to Definition \ref{Def:label}. First the base case: either $x$ is a leaf or all its children are leaves and we are in Case 0 and the label is assigned according to Def. \ref{Def:label}. Otherwise, observe the decision tree in Figure \ref{diagram}. It follows from Def. \ref{Def:label}, $k$, $N_k$ and $t_k$ that cases 1 up to 7 of Prop. \ref{Prop: Case} corresponds to cases 1 up to 7 in the decision tree - we mention this correspondence in the below - and this proves that exactly one case applies to every rooted tree.
The following facts simplify the case analysis: 
$\lm(T_r[x])$ is equal to either $k$ or $k+1$, and since no subtree rooted in a child of $x$ has LMIM-width $k+1$ there cannot be any $(k+1)$-critical node in $T_r[x]$, therefore if $\lm(T_r[x]) = k+1$, $T_r[x]$ is always a type 1 tree and by Theorem \ref{Thm: Classification theorem} it must contain a node $v$ such that $D_{T_r[x]}(v,k) >= 3$. This node must either be a $k$-critical node in a rooted subtree of $T_r[x]$, or $x$ itself. We go through the cases 1 to 7 in order.\\
Note that in Cases 1, 2, and 3 the condition 'Every label in $N_k$ is simple and has $last\_type$ equal to $t.1$ or $t.0$' means there are no $k$-critical nodes in any subtree of $T_r[x]$, because every $T_r[v]$ for $v \in Child(x)$ is either of type 1 or has \LM $< k$:\\
\textbf{Case 1:} By definition of $t_k$, $D_{T_r[x]}(x,k) \leq 1$. Therefore, $\lm(T_r[x]) = k$, and $T_r[x]$ is a type 1 tree.\\
\textbf{Case 2:} By definition of $t_k$, $D_{T_r[x]}(x,k) = 2$, and no other nodes are $k$-critical, therefore $\lm(T_r[x]) = k$. But now $x$ is $k$-critical in $T_r[x]$ so $T_r[x]$ is a type 2 tree.\\
\textbf{Case 3:} By definition of $t_k$, $D_{T_r[x]}(x,k) = 3$ and $\lm(T_r[x]) = k+1$.\\
For the remaining Cases 4, 5, 6 and 7, some $T_r[v]$ for $v \in Child(x)$ has \LM $k$ and is of type 2, 3 or 4, so at least one $k$-critical node exists in some subtree of $T_r[x]$:\\
\textbf{Case 4:} There is a $k$-critical node $u_k$ in some $T_r[v_i]$ (not of type 1), and some other $v_j$ has $\lm(T_r[v_j]) = k$ (because $|N_k| \geq 2$). Now observe $w$ the parent of $u_k$. The dangling tree $T_r[x]\backslash T_r[w]$ is a supertree of $T_r[v_j]$ and thus has \LM $\geq k$. Therefore $w$ is a $k$-neighbour of $u_k$ and by Theorem \ref{Thm: Classification theorem} $\lm(T_r[x]) = k+1$.\\
\textbf{Case 5:} $x$ has only one child $v$ with $\lm(T_r[v]) = k$, and $v$ is itself $k$-critical ($T_r[v]$ is type 2). $x$ cannot be a $k$-neighbour of $v$ in the unrooted $T_r[x]$, because every dangling tree from $x$ is some $T_r[v_i], v_i \neq v$ of $x$, which we know has \LM $< k$. Since no other node in $T$ is $k$-critical, $\lm(T_r[x]) = k$, and since $v$, a child of $x$, is $k$-critical in $T_r[x]$, $T_r[x]$ is a type 3 tree.\\
\textbf{Case 6:} $x$ has only one child $v$ with $\lm(T_r[v]) = k$, and there is a $k$-critical node $u_k$ with parent $w$ – neither of which are equal to $x$ – in $T_r[v]$ ($T_r[v]$ is a type 3 or type 4 tree). Moreover, no tree rooted in another child of $w$, apart from $u_k$, can have \LM $\geq k$, since this would imply $D_{T_r[v]}(u_k,k) = 3$ and thus $\lm(T_r[v]) > k$; nor can $T_r[x]\backslash T_r[w]$ have \LM $= k$, since then we would have $k$ in $label(T_r[x]\backslash T_r[w])$ disagreeing with the condition of Case 6. Therefore $D_{T_r[x]}(u,k) = 2$, and $\lm(T_r[x]) = k$. $T_r[x]$ is thus a type 4 tree and the label is assigned according to the definition.\\
\textbf{Case 7:} $T_r[v]$, $u_k$ and $w$ are as described in Case 6. But here, $\lm(T_r[x]\backslash T_r[w]) = k$ (since the condition says that $k$ is in its label), and thus $w$ is a $k$-neighbour of its child $u_k$ and by Theorem \ref{Thm: Classification theorem} $\lm(T_r[x]) = k+1$.\\
We conclude that $label(T_r[x])$ has been assigned the correct value in all possible cases.
\end{proof}

\begin{center}
\begin{figure}[ht!]
\centering
\includegraphics[width=120mm]{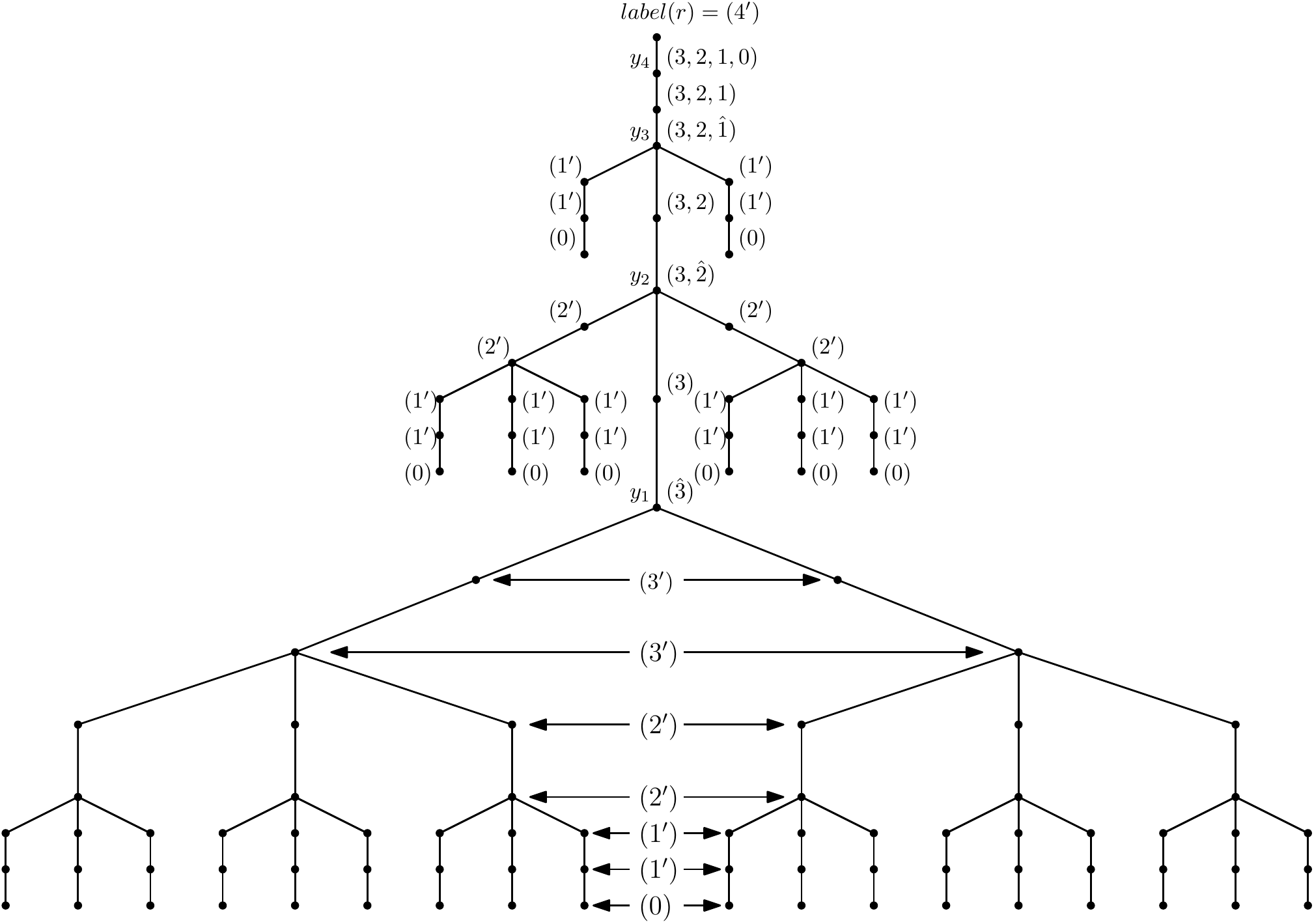}
\caption{A rooted tree of LMIM-width 4 with labels of subtrees. We explain the labels $(3,t.2),(3,t.3),(3,2,t.2)$ assigned to subtrees rooted at the nodes we call $a,b,c$, with parent$(a)=b$ and parent$(b)=c$. The sub-tree rooted at $a$, with label $(3,t.2)$ has precisely two children that have a child-tree each of LMIM-width 3, hence $a$ is $3$-critical and it is a type $2$ tree (Case 2 of Prop. \ref{Prop: Case}).  The sub-tree rooted at $b$, labelled $(3,t.3)$, is thus the parent of a $3$-critical node, and so it is of type $3$ (Case 5 of Prop. \ref{Prop: Case}). The sub-tree rooted at $c$ with label $(3,2,t.2)$ has maximum LMIM-width of a child-tree being 3, and it has a 3-critical node $a$ which is neither $c$ nor a child of $c$, so it is of type 4 (Case 6 of Prop. \ref{Prop: Case}); and moreover the subtree $T_r[c] \setminus T_r[a]$ has LMIM-width 2 with node $c$ as $2$-critical so it is of type 2 (Case 2 of Prop. \ref{Prop: Case}), and the label of $T_r[c]$ becomes $3 \oplus (2,t.2)$.}
\label{treewithlabels}
\end{figure}
\end{center}

\section{Computing LMIM-width of Trees and Finding a Layout} \label{sec:comp}

The subroutine underlying Prop. \ref{Prop: Case} will be used in a bottom-up algorithm that starts out at the leaves and works its way up to the root, computing labels of subtrees $T_r[x]$. However, in two cases (Case 6 and 7) we need the label of $T_r[x] \backslash T_r[w]$, which is not a complete subtree rooted in any node of $T_r$. Note that the label of $T_r[x] \backslash T_r[w]$ is again given by a (recursive) call to Prop. \ref{Prop: Case} and is then stored as a suffix of the complex label of $T_r[x]$. We will compute these labels by iteratively calling Prop. \ref{Prop: Case} (substituting the recursion by iteration).
We first need to carefully define the subtrees involved when dealing with complex labels.

From the definition of labels it is clear that only type 4 trees lead to a complex label. In that case we have a tree $T_r[x]$ of LMIM-width $k$ and a $k$-critical node $u_k$ that is neither $x$ nor a child of $x$, and the recursive definition gives $label(T_r[x]) = k \oplus label(T_r[x]\backslash T_r[w])$ for $w$ the parent of $u_k$.
Unravelling this recursive definition, this means that if $label(T_r[x]) = (a_1,\ldots,a_p,last\_type)$, we can define a list of nodes $(w_1,\ldots,w_{p-1})$ where $w_i$ is the parent of an $a_i$-critical node in $T_r[x]\backslash(T_r[w_1]\cup\ldots\cup T_r[w_{i-1}])$. We expand this list with $w_p = x$, such that there is one node in $T_r[x]$ corresponding to each number in $label(T_r[x])$, and $T_r[x]\backslash(T_r[w_1]\cup\ldots\cup T_r[w_p]) = \emptyset$.

Now, in the first level of a recursive call to Prop. \ref{Prop: Case} the role of $T_r[x]$ is taken by $T_r[x]\backslash T_r[w_1]$, and in the next level it is taken by $(T_r[x]\backslash T_r[w_1])\backslash T_r[w_2]$ etc. The following definition gives a shorthand for denoting these trees.

\begin{definition} 
Let $x$ be a node in $T_r$, $label(T_r[x])=(a_1,a_2,\dots,a_p,last\_type)$ and the corresponding list of vertices $(w_1,\dots, w_p)$ is as we describe in the above text. For any non-negative integer $s$, the tree $\mathbf{T_r[x,s]}$ is the subtree of $T_r[x]$ obtained by removing all trees $T_r[w_i]$ from $T_r[x]$, where $a_i\geq s$. In other words, if $q$ is such that $a_q \geq s > a_{q+1}$, then
$T_r[x,s] = T_r[x]\backslash(T_r[w_1]\cup T_r[w_2]\cup\ldots\cup T_r[w_q])$
\end{definition}

\begin{remark} \label{Rmk: Subtrees by labels}
Some important properties of $T_r[x,s]$ are the following. Let $T_r[x,s]$, $label(T_r[x,s])$, $(w_1,\ldots,w_p)$ and $q$ as in the definition. Then
\begin{enumerate}
\item if $s > a_1$, then $T_r[x,s] = T_r[x]$  \label{Rmk: same}
   \item  $label(T_r[x,s]) = (a_{q+1},\ldots,a_p,last\_type)$ \label{Rmk: Label-subset}
   \item $\lm(T_r[x,s]) = a_{q+1} < s$  \label{Cor: Subtree LMW-bound}
   \item $\lm(T_r[x,s+1]) = s$ \emph{if and only if} $s \in label(T_r[x])$ \label{Cor: Max Subtree LMW}
    \item $T_r[x,s+1] \neq T_r[x,s]$ if and only if $s \in label(T_r[x])$ \label{Rmk: Different s and s+1}
\end{enumerate}
\end{remark}

\begin{proof}
These follow from the definitions, maybe the last one requires a proof:\\
\emph{Backward direction:} Let $s = a_q$ for some $1 \leq q \leq p$. Then 
$T_r[x,s+1] = T_r[x]\backslash(T_r[w_1]\cup\ldots\cup T_r[w_{q-1}])$ and $T_r[x,s] = T_r[x]\backslash(T_r[w_1]\cup\ldots\cup T_r[w_q])$. These two trees are clearly different.\\
\emph{Forward direction:} Let $T_r[x,s] = T_r[x]\backslash(T_r[w_1]\cup\ldots\cup T_r[w_q])$ and $T_r[x,s+1] = T_r[x]\backslash(T_r[w_1]\cup\ldots\cup T_r[w_{q'}])$ with $q' < q$ and $a_{q'} > a_q$ (because numbers in a label are strictly descending). $a_q < s+1$ and $a_q \geq s$, ergo $a_q = s$.
\end{proof} 

Note that for any $s$ the tree $T_r[x,s]$ is defined only after we know $label(T_r[x])$. In the algorithm, we compute $label(T_r[x])$ by iterating over increasing values of $s$ (until $s > \lm(T_r[x])$ since by Remark \ref{Rmk: Subtrees by labels}.\ref{Rmk: same} we then have $T_r[x,s] = T_r[x]$) and we could hope for a loop invariant saying that we have correctly computed $label(T_r[x,s])$. However, $T_r[x,s]$ is only known once we are done. Instead, each iteration of the loop will correctly compute the label of the following subtree called $T_{union}[x,s]$, which is not always equal to $T_r[x]$, but importantly for $s > \lm(T_r[x])$, we will have $T_{union}[x,s] = T_r[x,s] = T_r[x]$.

\begin{definition} \label{Def: Union}
Let $x$ be a node in $T_r$ with children $v_1,\ldots,v_d$. $T_{union}[x,s]$ is then equal to the tree induced by $x$ and the union of all $T_r[v_i,s]$ for $1 \leq i \leq d$. More technically, $T_{union}[x,s] = T_r[V']$ where $V' = x\cup V(T_r[v_1,s])\cup\ldots\cup V(T_r[v_d,s])$. 
\end{definition}

Given a tree $T$, we find its \LM by rooting it in an arbitrary node $r$, and computing labels by processing $T_r$ bottom-up. The answer is given by the first element of $label(T_r[r])$, which by definition is equal to $\lm(T)$. At a leaf $x$ of $T_r$ we initialize by $label(T_r[x]) \gets (0,t.0)$, and at a node $x$ for which all children are leaves we initialize by $label(T_r[x]) \gets (1,t.0)$, according to Definition \ref{Def:label}. When reaching a higher node $x$ we compute label of $T_r[x]$ by calling function \textsc{MakeLabel}$(T_r,x)$.

\begin{algorithmic}
\Function{MakeLabel}{$T_r$, $x$} 
\Comment{finds $cur\_label=label(T_r[x])$}
	\State $cur\_label \gets (0,t.0)$
	\Comment{This is $label(T_{union}[x,0])$}
   \State $\{v_1,\ldots,v_d\} = $ children of $x$
	\If{$0 \in label(T_r[v_i])$ for some $i$}
		\State $cur\_label \gets (1,t.0)$
		\Comment{This is then $label(T_{union}[x,1])$}
	\EndIf
	\For{$s \gets 1,\max_{i=1}^d\{\text{first element of } label(T_r[v_i])\}$}
		\State $\{l_1',\ldots,l_d'\} = \{label(T_r[v_i,s+1])\ |\ 1 \leq i \leq d\}$
		\State $N_s = \{v_i\ |\ 1 \leq i \leq d,\ s \in l_i'\}$
		\State $t_s = |\{v_i\ |\ v_i \in N_s,\ v_i$ has child $u_j$ s.t. $s \in label(T_r[u_j,s+1])\}|$
		\If{$|N_s| > 0$}
			\State $case \gets$ the case from Prop. \ref{Prop: Case} applying  to $s$, $\{l_1',\ldots,l_d'\}$, $N_s$ and $t_s$
			\State $cur\_label \gets$ as given by $case$ in Prop. \ref{Prop: Case} ($s  \oplus  cur\_label$ if Case 6)
		\EndIf
	\EndFor
\EndFunction
\end{algorithmic}
\begin{center}
\begin{figure}[ht!]
\centering
\includegraphics[width=100mm]{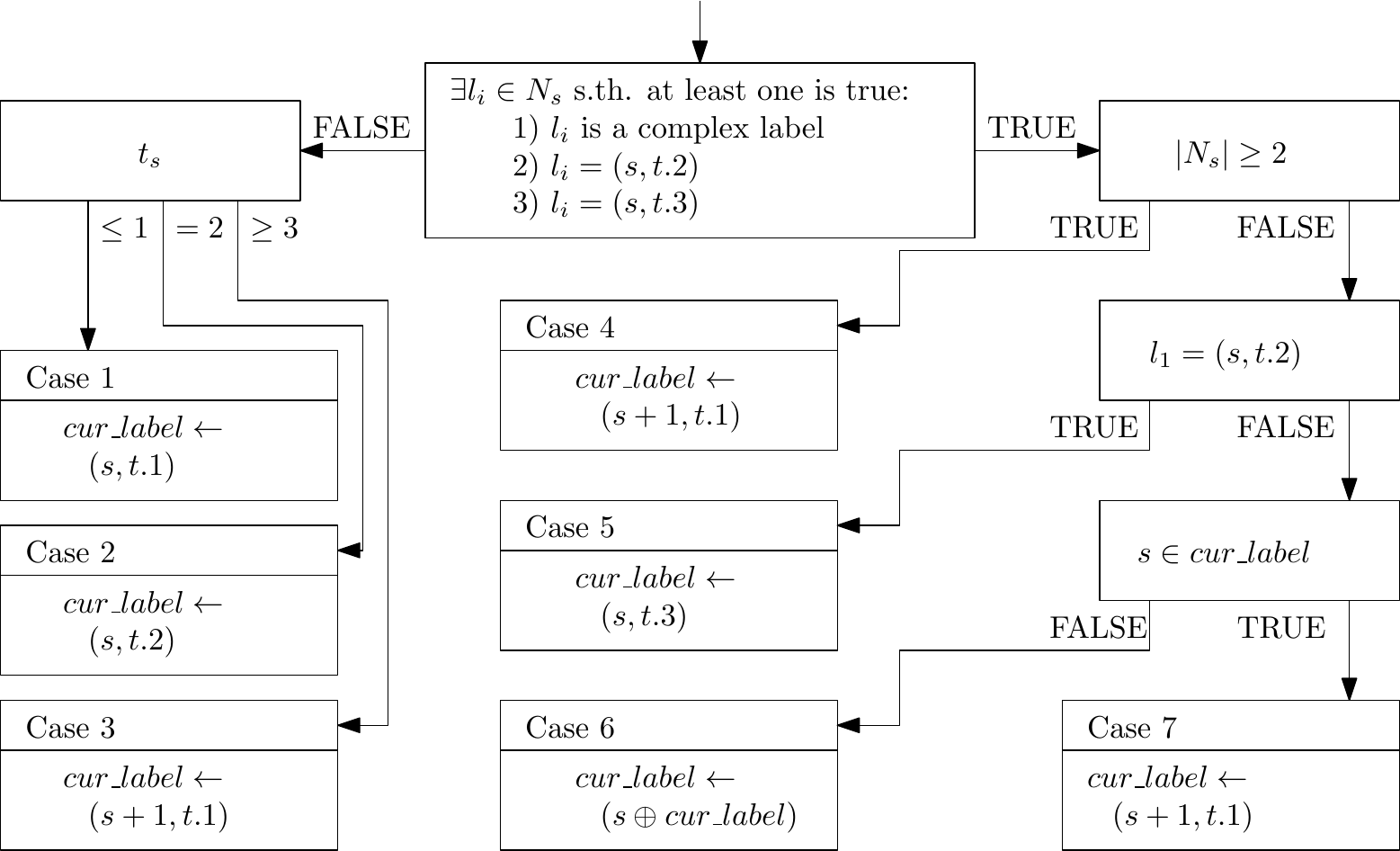}
\caption{The same decision tree as shown in Prop. \ref{Prop: Case}, but adapted to \textsc{MakeLabel}}
\label{fig:dtree}
\end{figure}
\end{center}

\begin{lemma} \label{Lemma: MakeLabel correctness}
Given labels at descendants of node $x$ in $T_r$, \textsc{MakeLabel}($T_r,x$) computes $label(T_r[x])$ as the value of $cur\_label$.
\end{lemma}

%

\begin{proof}
Assume that $x$ has the children $v_1,\ldots,v_d$, and denote their set of labels as $L = \{l_1,\ldots,l_d\}$.
\textsc{MakeLabel} keeps a variable $cur\_label$ that is updated maximally $k$ times in a for loop, where $k$ is the biggest number in any label of children of $x$. The following claim will suffice to prove the lemma, since for $s > \lm(T_r[x])$, we have $T_{union}[x,s] = T_r[x]$..\\

Claim: At the end of the $s$'th iteration of the for loop the value of $cur\_label$ is equal to $label(T_{union}[x,s+1])$. \\

\emph{Base case}: We have to show that before the first iteration of the loop we have $cur\_label=label(T_{union}[x,1])$. If some label $l_i\in L$ has 0 as an element then $T_{union}[x,1]$ is isomorphic to a star with $x$ as the center and $v_i$ as a leaf. By Prop. \ref{Prop: Case}, in this case $label(T_{union}[x,1])=(1,t.0)$ and this is what $cur\_label$ is initialized to.
If no $l_i\in L$ has $0$ as an element, then by Remark \ref{Rmk: Subtrees by labels}.\ref{Rmk: Different s and s+1} $T_{union}[x,1]=T_{union}[x,0]$ which by definition is the singleton node $x$ and by Prop. \ref{Prop: Case} the label of this tree is $(0,t.0)$ and this is what $cur\_label$ is initialized to.

\emph{Induction step}: We assume $cur\_label = label(T_{union}[x,s])$ at the start of the $s$'th iteration of the for loop and show that at the end of the iteration, $cur\_label = label(T_{union}[x,s+1])$.\\
The first thing done in the for loop is the computation of
$\{l_i'\ |\ 1 \leq i \leq d,\ l_i' = label(T_r[v_i,s+1])\}$. By Remark \ref{Rmk: Subtrees by labels}.\ref{Rmk: Label-subset}, $label(T_r[v_i,s+1]) \subseteq label(T_r[v_i])$ for all $i$, therefore $l_1',\ldots,l_d'$ are trivial to compute.
The second thing done is to set $N_s$ as the set of all children of $x$ whose labels contain $s$, and $t_s$ as the number of nodes in $N_s$ that themselves have children whose labels contain $s$.
Let us first look at what happens when $|N_s| = 0$:\\
By Remark \ref{Rmk: Subtrees by labels}.\ref{Rmk: Different s and s+1}, for every child $v_i$ of $x$, $T_r[v_i,s+1] = T_r[v_i,s]$ if $s \not \in label(T_r[v_i])$. Therefore, if $|N_s| = 0$, then $T_{union}[x,s+1] = T_{union}[x,s]$, and from the induction assumption, $label(T_{union}[x,s+1]) = cur\_label$, and indeed when $|N_s|=0$ then iteration $s$ of the loop does not alter $cur\_label$. \\
Otherwise, we have $|N_s| > 0$ and make a call to the subroutine given by Prop. \ref{Prop: Case}, see the decision tree in Figure \ref{fig:dtree}, to compute $label(T_{union}[x,s+1])$ and argue first that the variables used in that call correspond to the variables used in Prop. \ref{Prop: Case} to compute $label(T_r[x])$. The correspondence is given in Table \ref{tabcor}.
\begin{table}[]\label{tabcor}
\begin{tabular}{|l|l|l|}
\hline
Proposition \ref{Prop: Case} & for loop iteration $s$& Explanation\\
\hline
$T_r[x], k$ &  $T_{union}[x,s+1], s$  & Tree needing label, max $\lm$ of children\\
$T_r[v_1],...,T_r[v_d]$ &  $T_r[v_i,s],...,T_r[v_d,s]$ & Subtrees of children \\
$l_1,...,l_d, N_k, t_k$ & $l_1',...,l_d',N_s,  t_s$ & Child labels, those with max, root comp. index\\
$label(T_r[x]\backslash T_r[w])$ & $cur\_label$ & This is also $label(T_{union}[x,s+1]\backslash T_r[w,s+1])$\\
 \hline
\end{tabular}
\end{table}
Most of these are just observations: $T_{union}[x,s+1]$ corresponds to $T_r[x]$ in Prop. \ref{Prop: Case}, and $T_r[v_1,s+1],\ldots,T_r[v_d,s+1]$ corresponds to $T_r[v_1],\ldots,T_r[v_d]$.
$\{l_i'\ |\ 1 \leq i \leq d,\ l_i' = label(T_r[v_i,s+1])\}$ correspond to $\{label(T_r[v])\ |\ v \in Child\}$ in Prop. \ref{Prop: Case}.
$N_s$ is defined in the algorithm so that it corresponds to $N_k$ in Prop. \ref{Prop: Case}. Since $|N_s| > 0$, some $v_i$ has $s$ in its label $l_i'$.
By Remark \ref{Rmk: Subtrees by labels}.\ref{Cor: Subtree LMW-bound} and \ref{Rmk: Subtrees by labels}.\ref{Cor: Max Subtree LMW}, we can infer that $s$ is the maximum \LM of all $T_r[v_i,s+1]$, therefore $s$ corresponds to $k$ in Proposition \ref{Prop: Case}.\\
It takes a bit more effort to show that $t_s$ computed in iteration $s$ of the for loop corresponds to $t_k = D_{T_r[x]}(x,k)$ in Prop. \ref{Prop: Case} – meaning we need to show that $t_s = D_{T_{union}[x,s+1]}(x,s)$. Consider $v_i$, a child of $x$. In accordance with \textsc{MakeLabel} we
say that $v_i$ \emph{contributes to} $t_s$ if $v_i \in N_s$ and $v_i$ has a child $u_j$ with $s$ in its label. We thus need to show that $v_i$ contributes to $t_s$ if and only if $v_i$ is an $s$-neighbour of $x$ in $T_{union}[x,s+1]$.
Observe that by Remark \ref{Rmk: Subtrees by labels}.\ref{Cor: Max Subtree LMW}, $\lm(T_r[v_i,s+1]) = \lm(T_r[u_j,s+1]) = s$ if and only if $s$ is in the labels of both $T_r[v_i]$ and $T_r[u_j]$. If $s \not\in label(T_r[u_j,s+1])$, then $\lm(T_r[u_j,s+1]) < s$, and if this is true for all children of $v_i$, then $v_i$ is not an $s$-neighbour of $x$ in $T_{union}[x,s+1]$. If $s \not\in label(T_r[v_i,s+1])$, then $\lm(T_r[v_i,s+1]) < s$ and no subtree of $T_r[v_i,s+1]$ can have \LM $s$.
However, if $s \in label(T_r[u_j,s+1])$ and $s \in label(T_r[v_i,s+1])$ (this is when $v_i$ contributes to $t_s$), then $T_r[v_i,s+1] \cap T_r[u_j]$ must be equal to $T_r[u_j,s+1]$ and $T_r[u_j,s+1] \subseteq T_{union}[x,s+1]$, and we conclude that $v_i$ is an $s$-neighbour of $x$ in $T_{union}[x,s+1]$ if and only if $v_i$ contributes to $t_s$, so $t_s = D_{T_{union}[x,s+1]}(x,s)$. \\
Lastly, we show that if $T_{union}[x,s+1]$ is a Case 6 or Case 7 tree – that is, $|N_s| = 1$, and $T_r[v_1,s+1]$ is a type 3 or type 4 tree, with $w$ being the parent of an $s$-critical node – then the algorithm has $label(T_{union}[x,s+1]\backslash T_r[w,s+1])$ available for computation, indeed that this is the value of $cur\_label$. We know, by definition of label and Remark \ref{Rmk: Subtrees by labels}.\ref{Rmk: Different s and s+1} that $T_r[v_i,s+1]\backslash T_r[v_i,s] = T_r[w,s+1]$. But since $|N_s| = 1$, for every $j \neq i$, $T_r[v_j,s+1]\backslash T_r[v_j,s] = \emptyset$. Therefore $T_{union}[x,s+1]\backslash T_{union}[x,s] = T_r[w,s+1]$ and $T_{union}[x,s+1]\backslash T_r[w,s+1] = T_{union}[x,s]$. But by the induction assumption, $cur\_label = label(T_{union}[x,s])$. Thus $cur\_label$ corresponds to $label(T_r[x]\backslash T_r[w])$ in Prop. \ref{Prop: Case}.\\
We have now argued for all the correspondences in Table \ref{tabcor}. By that, we conclude from Prop. \ref{Prop: Case} and Definition \ref{Def: Union} and the inductive assumption that $cur\_label = label(T_{union}[x,s+1])$ at the end of the $s$'th iteration of the for loop in \textsc{MakeLabel}.
It runs for $k$ iterations, where $k$ is equal to the biggest number in any label of the children of $x$, and $cur\_label$ is then equal to $label(T_{union}[x,k+1])$. Since $k \geq \lm(T_r[v_i])$ for all $i$, by definition $T_r[v_i,k+1] = T_r[v_i]$ for all $i$, and thus $T_{union}[x,k+1] = T_r[x]$. Therefore, when \textsc{MakeLabel} finishes, $cur\_label = label(T_r[x])$.
\end{proof}

\begin{theorem}
Given any tree $T$, $\lm(T)$ can be computed in $\mathcal{O}(n\log(n))$-time.
\end{theorem}

\begin{proof}
We find $\lm(T)$ by bottom-up processing of $T_r$ and returning the first element of $label(T_r)$. After correctly initializating at leaves and nodes whose children are all leaves, we make a call to \textsc{MakeLabel} for each of the remaining nodes. Correctness follows by Lemma \ref{Lemma: MakeLabel correctness} and induction on the structure of the rooted tree.
For the timing we show that each call runs in $\mathcal{O}(log\ n)$ time.
For every integer $s$ from 1 to $m$, the biggest number in any label of children of $x$, which is $O(\log n)$ by Remark \ref{Rmk: log lmw}, the algorithm checks how many labels of children of $x$ contain $s$ (to compute $N_s$), and how many labels of grandchildren of $x$ contain $s$ (to compute $t_s$). The labels are sorted in descending order, therefore the whole loop 
goes only once through each of these labels, each of length $O(\log n)$.
Other than this, \textsc{MakeLabel} only does a constant amount of work. 
Therefore, \textsc{MakeLabel}$(T_r, x)$, if $x$ has $a$ children and $b$ grandchildren, takes time proportional to $O(\log n)(a+b)$.
As the sum of the number of children and grandchildren over all nodes of $T_r$ is $O(n)$ we conclude that the total runtime to compute  $\lm(T)$ is $\mathcal{O}(n \cdot log\ n)$.
\end{proof}


\begin{theorem} \label{thm:lay}
A layout of LMIM-width $\lm(T)$ of a tree $T$ can be found in $\mathcal{O}(n\cdot log\ n)$-time.
\end{theorem}

\begin{proof}
Given $T$ we first run the algorithm computing $\lm(T)$ by finding labels of all nodes and various subtrees.
Given $T$ we first run the algorithm computing $\lm(T)$ finding the label of every full rooted subtree in $T_r$.
We give a recursive layout-algorithm that uses these labels in tandem with \textsc{LinOrd} presented in the Path Layout Lemma. We call it on a rooted tree where labels of all subtrees are known. For simplicity we call this rooted tree $T_r$ even though in recursive calls this is not the original root $r$ and tree $T$. 
The layout-algorithm goes as follows:\\
\textbf{1)} Let $\lm(T_r)=k$ and find a path $P$ in $T_r$ such that all trees in $T_r \backslash N[P]$ have \LM $< k$. The path depends on the type of $T_r$ as explained in detail below.\\
\textbf{2)} Call this  layout-algorithm recursively on every rooted tree in $T_r\backslash N[P]$ to obtain linear layouts; to this end, we need the correct label for every node in these trees.\\
\textbf{3)} Call \textsc{LinOrd} on $T_r$, $P$ and the layouts provided in step 2.\\

Every tree in the forest $T\backslash N[P]$ is equal to a dangling tree $T\langle v,u\rangle$, where $v$ is a neighbour of some $x\in P$.\\
We observe that if $\lm(T) = k$, then by definition $\lm(T\langle v,u\rangle) = k$ if and only if $v$ is a $k$-neighbour of $x$. It follows that every tree in $T\backslash N[P]$ has LMIM-width at most $k-1$ if and only if no node in $P$ has a $k$-neighbour that is not in $P$. We use this fact to show that for every type of tree we can find a satisfying path in the following way:\\

\emph{Type 0 trees:} Choose $P = (r)$. Since $T\backslash N[r] = \emptyset$ in these trees, this must be a satisfying path.\\
\emph{Type 1 trees:} These trees contain no $k$-critical nodes, which by definition means that for any node $x$ in $T_r$, at most one of its children is a $k$-neighbour of $x$. Choose $P$ to start at the root $r$, and as long as the last node in $P$ has a $k$-neighbour $v$, $v$ is appended to $P$. This set of nodes is obviously a path in $T_r$. No node in $P$ can possibly have a $k$-neighbour outside of $P$, therefore all connected components of $T\backslash N[P]$ have LMIM-width $\leq k-1$. Furthermore, all components of $T-N[P]$ are full rooted sub-trees of $T_r$ and so the labels are already known.\\
\emph{Type 2 trees:} In these trees the root $r$ is $k$-critical. We look at the trees rooted in the two $k$-neighbours of $r$, $T_r[v_1]$ and $T_r[v_2]$. By Remark \ref{Rmk: single k-crit} these must both be Type $1$ trees, and so we find paths $P_1,P_2$ in $T_r[v_1]$ and $T_r[v_2]$ respectively, as described above. Gluing these paths together at $r$ we get a satisfying path for $T_r$, and we still have correct labels for the components $T\backslash N[P]$.\\
\emph{Type 3 trees:} In these trees, $r$ has exactly one child $v$ such that $T_r[v]$ is of type $2$ and none of its other children have LMIM-width $k$. We choose $P$ as we did above for $T_r[v]$. $r$ is clearly not a $k$-neighbour of $v$, or else $D_T(v,k) = 3$. Every other node in $P$ has all their neighbours in $T_r[v]$. Again, every tree in $T\backslash N[P]$ is a full rooted subtree, and every label is known.\\
\emph{Type 4 trees:} In these trees, $T_r$ contains precisely one node $w \neq r$ such that $w$ is the parent of a $k$-critical node, $x$. This $w$ is easy to find using the labels, and clearly the tree $T_r[w]$ is a type $3$ tree with LMIM-width $k$. We find a path $P$ that is satisfying in $T_r[w]$ as described above. $w$ is still not a $k$-neighbour of $x$, therefore $P$ is a satisfying path. In this case, we have one connected component of $T\backslash N[P]$ that is not a full rooted subtree of $T_r$, that is $T_r\backslash T_r[w]$. Thus for every ancestor $y$ of $w$ (the blue path in Figure \ref{fig:lay}) $T_r[y] \backslash T_r[w]$ is not a full rooted sub-tree either, and we need to update the labels of these trees. However, $T_r[y]\backslash T_r[w]$ is by definition equal to $T_r[y,k]$, whose label is equal to $label(T_r[y])$ without its first number. Thus we quickly find the correct labels to do the recursive call.
\end{proof}

\begin{center}
\begin{figure}[ht!]
\centering
\includegraphics[width=120mm]{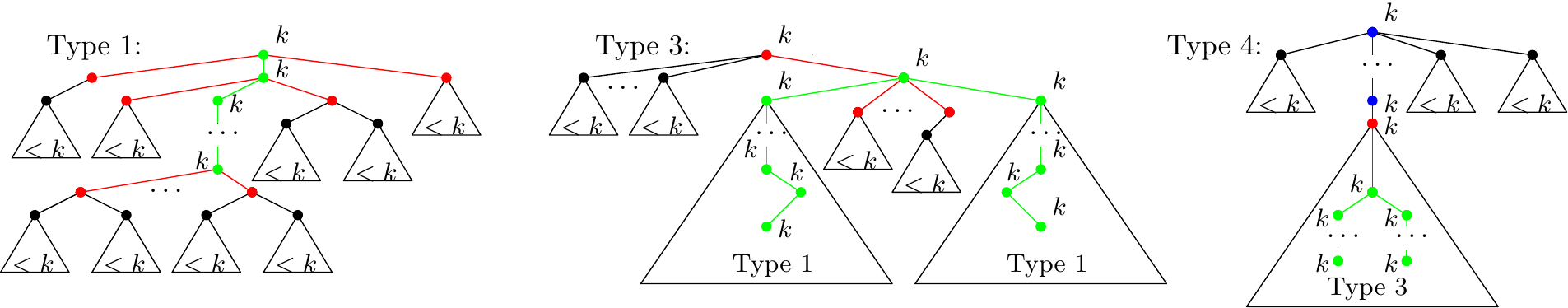}
\caption{The path $P$ in green for the proof of Theorem \ref{thm:lay}.}
\label{fig:lay}
\end{figure}
\end{center}

\section{Conclusion}

We have given an $O(n \log n)$ algorithm computing the LMIM-width and an optimal layout of an $n$-node tree. 
This is the first graph class of LMIM-width larger than $1$ having a polynomial-time algorithm computing LMIM-width and thus constitutes an important step towards a better understanding of LMIM-width. 
Indeed, for the development of FPT algorithms computing tree-width and path-width of general graphs, one could argue that the algorithm of 
\cite{EllisST94} computing optimal path-decompositions of a tree in time $O(n \log n)$ was a stepping stone. 
The situation is different for MIM-width and LMIM-width, as it is W-hard to compute these parameters \cite{SaetherV16}, but it is similar in the sense that
our objective has been to achieve an understanding of how to take a graph and assemble a decomposition of it, in this case a linear one,  such that it has cuts of low MIM. To achieve this objective a polynomial-time algorithm for trees has been our main goal. 

\bibliography{ref} 
\bibliographystyle{plain}

\end{document}